\documentclass[11pt,a4paper]{article}
\usepackage[T1]{fontenc}
\usepackage[utf8]{inputenc}
\usepackage{lmodern,amsmath,amssymb,amsthm,mathtools,microtype}
\usepackage[a4paper,margin=25mm]{geometry}
\usepackage{enumitem,booktabs,tikz,aliascnt}
\usetikzlibrary{arrows.meta,positioning}
\usepackage{caption,authblk}
\usepackage[hidelinks]{hyperref}
\usepackage[nameinlink,noabbrev,capitalise]{cleveref}
\crefname{appendix}{Appendix}{Appendices}
\Crefname{appendix}{Appendix}{Appendices}
\setlist{itemsep=1pt,topsep=4pt,parsep=0pt,leftmargin=2em}
\allowdisplaybreaks[1]
\numberwithin{equation}{section}
\newtheorem{theorem}{Theorem}[section]
\crefname{theorem}{Theorem}{Theorems}
\newaliascnt{lemma}{theorem}
\newtheorem{lemma}[lemma]{Lemma}
\aliascntresetthe{lemma}
\newaliascnt{proposition}{theorem}

\aliascntresetthe{proposition}
\newaliascnt{corollary}{theorem}

\aliascntresetthe{corollary}
\DeclareMathOperator{\Var}{Var}
\DeclareMathOperator{\bit}{bit}
\newcommand{\E}{\mathbb E}
\newcommand{\Pp}{\mathbb P}
\newcommand{\1}{\mathbf 1}
\newcommand{\cE}{\mathcal E}
\newcommand{\TV}{\mathrm{TV}}
\newcommand{\enc}{T_{\mathrm{enc}}}
\newcommand{\Cp}{C_{\mathrm{pair}}}
\newcommand{\Jout}{J_{\mathrm{out}}}
\newcommand{\Bflat}{B^{\mathrm{flat}}}
\newcommand{\norm}[1]{\lVert #1\rVert}

\newcommand{\doi}[1]{\href{https://doi.org/#1}{\nolinkurl{doi:#1}}}
\hypersetup{pdftitle={A simplified proof of the FPRAS for counting perfect matchings},pdfauthor={Chenghua Liu, Boning Meng},pdfsubject={Cycle exchanges, marked ultrametric completion, and approximate counting}}
\title{A simplified proof of the FPRAS for counting perfect matchings}
\author[1]{Chenghua Liu}
\author[2]{Boning Meng}
\affil[1]{Institute of Software, Chinese Academy of Sciences, Beijing, China}
\affil[2]{University of Regensburg, Regensburg, Germany}
\affil[ ]{\href{mailto:liuch.russell@gmail.com}{\texttt{liuch.russell@gmail.com}},
\href{mailto:mengboning2013@gmail.com}{\texttt{mengboning2013@gmail.com}}}
\date{}
\begin{document}
\makeatletter
\begin{center}
{\LARGE\@title\par}
\medskip
{\large\@author\par}
\end{center}
\makeatother
\begin{abstract}
We give a simplified, self-contained proof of OpenAI's FPRAS for counting perfect matchings in general graphs. As in the original proof, we use a restricted Poincar\'e inequality for two matchings to sample weighted perfect matchings and estimate their total weight. We simplify the proof of this inequality by isolating parts of alternating cycles directly on the original vertex set and controlling the resulting error terms. This avoids the auxiliary vertices and tree constructions used in the original proof. The argument preserves the required parameter bounds and yields the same FPRAS guarantee.
\end{abstract}

\section{Introduction}\label{sec:intro}
A perfect matching of a graph is a set of edges that covers every vertex exactly once. Finding one and counting all of them have very different complexity. Edmonds~\cite{Edmonds} gave a polynomial-time algorithm for maximum matching in general graphs, whereas Valiant~\cite{Valiant} proved that counting perfect matchings is \#P-complete even in bipartite graphs. In the bipartite case, the count is the permanent of the bipartite adjacency matrix. Planar graphs admit polynomial-time exact counting through Kasteleyn's Pfaffian method~\cite{Kasteleyn}. For unrestricted graphs, can the count be approximated to prescribed relative error in randomized polynomial time?

Approximate counting is closely tied to almost-uniform sampling through the reductions of Jerrum, Valiant, and Vazirani~\cite{JVV} for self-reducible counting problems. Jerrum and Sinclair~\cite{JS} gave a fully polynomial randomized approximation scheme (FPRAS) for counting all matchings in an arbitrary graph, summed over all cardinalities. For counting perfect matchings, their method also allows near-perfect matchings, which leave exactly two vertices unmatched. Its guarantee requires a polynomial bound on the ratio of near-perfect to perfect matchings. Otherwise, rejection from the enlarged state space may succeed with exponentially small probability. Jerrum, Sinclair, and Vigoda~\cite{JSV} removed this restriction for bipartite graphs, obtaining an FPRAS for the permanent of every nonnegative matrix. They reweighted the near-perfect classes according to their unmatched vertices and learned these weights during annealing. Bez\'akov\'a, \v{S}tefankovi\v{c}, Vazirani, and Vigoda~\cite{BSVV} subsequently improved the cooling schedule and mixing analysis.

The extension to general graphs encounters concrete obstructions. \v{S}tefankovi\v{c}, Vigoda, and Wilmes~\cite[full version, Theorem 2.2]{SVW} constructed graphs on which every chain with the transitions in~\cite{JSV} and weights depending only on the unmatched vertices either mixes exponentially slowly or assigns exponentially small stationary mass to perfect matchings. Reweighting these classes alone therefore cannot repair the chain. Blossoms, the odd-cycle structures used in Edmonds' algorithm, do not give an immediate remedy: uniformly sampling all blossoms through a specified unmatched vertex is intractable unless $\mathrm{NP}=\mathrm{RP}$, while restricting to minimum blossoms leaves the bottleneck~\cite[full version, Section 3]{SVW}. Their recursion based on the Gallai--Edmonds decomposition handles further structured cases and the slow-mixing examples, but can take exponential time in general~\cite[full version, Section 4]{SVW}.

Restricting the state space to perfect matchings raises a different difficulty. Two-edge switches are not even irreducible on a chordless even cycle of length at least six. On quasimonotone graphs, Dyer and M\"uller~\cite{DM} proved rapid mixing of the two-edge switch chain, but they also exhibited exponentially slow mixing on examples in a larger hereditary class. A sampler based on two-edge switches therefore needs both a connectivity argument and a mixing bound.

Other approaches obtain approximation guarantees under additional assumptions on the graph or its edge weights. When a graph is regular and has sufficiently strong expansion, Ebrahimnejad, Nagda, and Oveis Gharan~\cite{ENO} bound the ratio of near-perfect to perfect matchings. When its degree is also bounded, this yields a polynomial-time sampler. Analytic methods provide a different route through hafnians, the weighted perfect-matching sums for general graphs. When all off-diagonal entries lie in $[\theta,1]$ for fixed $\theta>0$, Barvinok~\cite{Barvinok} gives deterministic quasipolynomial-time relative approximation. When the support has minimum degree at least $(1/2+\gamma)n$ and nonzero weights lie in $[\theta,1]$, with fixed $0<\gamma<1/2$ and $0<\theta\le1$, Yi~\cite{Yi} gives a deterministic fully polynomial-time approximation scheme. An FPRAS for general graphs must also handle arbitrary sparse graphs without any expansion assumption.

Mixing analyses also developed beyond the original matching chains. Sinclair's multicommodity-flow bounds~\cite{Sinclair} relate mixing to the load placed on chain transitions. In the bipartite setting, recent improvements for the permanent use restricted Poincar\'e inequalities and energy-based flows~\cite{CVY,CGVY}. For arbitrary graphs, OpenAI's September 2026 manuscript~\cite{OpenAI} obtains an FPRAS for counting perfect matchings by combining a two-matching energy comparison with replication, completion, and marked-edge estimators. We give a simplified proof of their result, including its worst-case bit guarantee.

\begin{theorem}[OpenAI~\cite{OpenAI}]\label{thm:main}
There is a randomized algorithm which, given a finite simple graph $G$ and rational $0<\varepsilon,\delta<1$, returns a nonnegative rational $\widehat Z$ such that
\[
 \Pp\bigl[(1-\varepsilon)Z(G)\le\widehat Z\le(1+\varepsilon)Z(G)\bigr]\ge1-\delta,
\]
where $Z(G)$ is the number of perfect matchings of $G$. Its worst-case bit running time and number of fair random bits are polynomial in the total input length, $\varepsilon^{-1}$, and $\log(1/\delta)$.
\end{theorem}

The simplification concerns the comparison argument in Sections 4--6 of~\cite{OpenAI}. We carry it out directly on the original vertex set and retain OpenAI's framework for converting the comparison into sampling and counting algorithms. We give the required arguments below.

To describe this framework, let $n$ be even and assign each edge $e$ of the complete graph $K_n$ a positive weight $a_e$, also called its activity. A perfect matching has weight equal to the product of its edge weights. Its sampling probability is proportional to this product:
\begin{equation}\label{eq:notation}
 w_a(M)=\prod_{e\in M}a_e,\qquad Z_a=\sum_Mw_a(M),\qquad
 \pi_a(M)=w_a(M)/Z_a,
\end{equation}
where $M$ ranges over perfect matchings and $Z_a$, the total weight, is the partition function. The counting algorithm starts with every edge weight equal to one. The initial total weight can then be computed directly. The algorithm gradually lowers the weights of edges absent from $G$, keeping the weights of graph edges equal to one. This procedure is called annealing. At the end, the total weight differs from $Z(G)$ by the small contribution of matchings using nonedges. Each successive ratio of total weights is estimated by averaging the multiplicative changes in matching weights over samples from the current distribution. Multiplying the estimated ratios by the known initial total weight gives the final estimate. The main algorithmic task is therefore to sample the weighted matching distributions along this sequence.

Sampling must remain efficient even when edge weights differ greatly. OpenAI's completion~\cite[proof of Theorem 3.1]{OpenAI} uses the following hierarchy. Keep only edges above a chosen weight threshold. As the threshold decreases, connected components merge into larger groups, forming a hierarchy. OpenAI use these connection thresholds to assign auxiliary weights to vertex pairs. Adding these to the original weights gives completed weights that are within a polynomial factor of hierarchical reference weights. They also show how to apply this completion while increasing the total matching weight by at most a constant factor.

For a distribution with this structure, consider two independent matchings $M_1,M_2\sim\pi_a$. Their differing edges form disjoint even cycles, with edges alternating between the two matchings. The comparison uses two operations: a \emph{two-edge switch} replaces two edges in $M_1$ by another pairing of their endpoints, and a \emph{cycle exchange} swaps the alternating patterns of $M_1$ and $M_2$ on an entire cycle. These operations are combined into a reversible random update. For a numerical quantity on matchings, we compare its overall fluctuation, measured by variance, with its mean squared change in one update. The restricted Poincar\'e inequality in \Cref{thm:pair} provides the comparison needed for the sampling construction. Its constant depends polynomially only on the number of vertices and the largest ratio of an edge weight to its corresponding reference weight, even when the reference weights differ greatly.

To prove \Cref{thm:pair}, we first reduce the variance bound to controlling changes in the chosen quantity under cycle exchanges. We call these changes cycle gains. We view each cycle as a polygon and divide it along the weight hierarchy. Adapting OpenAI's identity~\cite[Lemma 6.3]{OpenAI}, we express its gain through changes under two-edge switches and gains on smaller cycles, with errors from different exterior matchings. OpenAI's proof introduces auxiliary vertices for this decomposition. We instead isolate each smaller region directly from the original matching pair by one two-edge switch, so that a cycle exchange can act on it alone. The resulting weight loss has the same bound at every depth, and the construction uses only the original vertices.

We apply the same gain identity to each smaller cycle. Its gain appears in the larger cycle's expansion and on the left-hand side of its own identity. Direct isolation can change the exterior matching and hence the gain. Combining the identities leaves the difference between the two gains as an error. Subtracting the same mean from both gains expresses this error as a difference of centered gains. Cycle exchanges control these errors after summing over matching pairs. This adapts OpenAI's use of centering to control exterior errors~\cite[Section 6.4]{OpenAI}.

The expansion contains two-edge switch gains and centered gains controlled by cycle exchanges. Summing either type requires bounds on repeated contributions from different matching pairs. Adapting OpenAI's encoding~\cite[Lemma 6.4]{OpenAI}, we encode each contribution by an ordered matching pair for weight comparison and an auxiliary record. These data recover the original pair and the chosen expansion term. Since the record takes polynomially many values, only polynomially many contributions can give the same output pair. Comparing input and output weights then bounds their total probability mass. For a centered gain, the output pair contains the two alternating patterns on its cycle, allowing cycle exchanges to control the errors after summation. Finally, Cauchy--Schwarz bounds the square of the full gain expansion by a polynomial factor times the sum of squared terms. These estimates prove \cref{thm:pair}.

OpenAI's replication construction uses joint matching updates to turn the comparison into a polynomial-time approximate sampler for the completed weights. To relate these samples to the weights before completion, OpenAI mark sampled edges as real or virtual, distinguishing original and auxiliary weight contributions. The marked samples estimate the successive total-weight ratios used in annealing, completing the counting algorithm~\cite[Section 8.2]{OpenAI}.

We begin in \cref{sec:comparison} by proving the comparison under the hierarchical weight assumption. \Cref{sec:sampling} uses it to construct a sampler. \Cref{sec:completion} gives the completion and controls its increase in total matching weight, then applies the sampler to obtain marked matchings. \Cref{sec:counting} combines these ingredients to prove \cref{thm:main}. \Cref{app:implementation} supplies a finite-precision implementation with worst-case polynomial bounds on bit running time and random-bit use.

\section{A two-matching comparison}\label{sec:comparison}
We first prove the sampling comparison under the hierarchical weight assumption from the introduction. The completion in \cref{sec:completion} will verify this assumption.

Throughout, $\log$ denotes the natural logarithm. Let $\Omega$ be the set of perfect matchings of $K_n$ for even $n\ge4$, and $\pi=\pi_a$ the law in~\eqref{eq:notation} for positive activities $a$. Take independent $M_1,M_2\sim\pi$. Their differing edges form vertex-disjoint alternating even cycles, called \emph{discrepancy cycles}. Each has at least four vertices, so there are at most $K=\lfloor n/4\rfloor$.

\begin{samepage}
We formalize the hierarchical reference weights as a positive \emph{ultrametric similarity} $B$ on $K_n$: $B_{ij}\ge\min\{B_{ik},B_{kj}\}$ for distinct $i,j,k$. If $B_{ik},B_{kj}\ge t$, then $B_{ij}\ge t$. Edges with $B_e\ge t$ form disjoint cliques merging as $t$ decreases. Here $B\le a\le rB$ means $B_e\le a_e\le rB_e$ for all $e$.

\end{samepage}

Define a Markov chain on $\Omega\times\Omega$ with transition kernel $\mathsf K$. With probabilities $1/3$ each, propose a two-edge switch in $M_1$, hold, or propose a cycle exchange on a uniformly chosen discrepancy cycle of $(M_1,M_2)$. A two-edge switch chooses two distinct edges of $M_1$ uniformly and either alternative pairing with probability $1/2$. Accept the proposed matching $M'_1$ with Metropolis probability $\min\{1,\pi(M'_1)/\pi(M_1)\}$, and otherwise hold. Cycle exchanges are always accepted because they preserve the product of the two matching weights, and hence $\pi(M_1)\pi(M_2)$. If no discrepancy cycle exists, the cycle exchange branch holds. The proposals are symmetric, so $\mathsf K$ is reversible with stationary law $\pi\otimes\pi$.

For a reversible kernel $P$ with stationary law $\sigma$ and a real function $F$ on its state space, define
\[
 \cE_{\sigma,P}(F)=\frac12\sum_{x,y}\sigma(x)P(x,y)[F(x)-F(y)]^2.
\]
This is half the expected squared change in $F$ after one step at stationarity. Write $\cE=\cE_{\pi\otimes\pi,\mathsf K}$.

The following theorem plays the role of OpenAI's pair energy inequality~\cite[Theorem 6.1]{OpenAI}. Our proof uses direct isolation on the original vertex set, avoiding the edge subdivision and colored tree-bag analysis in OpenAI's proof.

\begin{theorem}\label{thm:pair}
Suppose $B\le a\le rB$, where $r\ge1$ and $B$ is a positive ultrametric similarity. Then, for all $f,g:\Omega\to\mathbb R$,
\begin{equation*}
 \Var_\pi(f)\le\Cp\,\cE\bigl(f(M_1)+g(M_2)\bigr),
 \qquad \Cp=\left\lceil(10n)^{80}r^9\right\rceil.
\end{equation*}
\end{theorem}
Changes in $g(M_2)$ may offset changes in $f(M_1)$ during cycle exchanges. The theorem still bounds $\Var_\pi(f)$ by the total energy for every $g$, with a constant polynomial only in $n$ and $r$. This form is used in the sampling construction of \cref{sec:sampling}.

We first bound the variance by the expected sum of squared changes in $f$ under cycle exchanges. We then bound this sum by the energies of two-edge switches and cycle exchanges.

\subsection{Cycle gains and cycle exchange energy}\label{sec:cycle-gains}
A cycle exchange can have different gains for different matchings outside the cycle. For each geometric even cycle $c$ of $K_n$ (a simple cycle viewed as an edge set), label its two alternating edge sets $P_c,Q_c$. Let $\pi_c$ be the matching law with activities $a$ on the vertices outside $c$, with weight one for the empty matching. For a perfect matching $U$ on these vertices, define
\begin{equation}\label{eq:cycle-data}
 \begin{gathered}
 d_c(U)=f(P_c\cup U)-f(Q_c\cup U),\qquad v_c=\Var_{\pi_c}(d_c),\\
 \theta_c=\pi(P_c\subseteq M)\pi(Q_c\subseteq M),
 \qquad V(f)=2\sum_c\theta_cv_c.
 \end{gathered}
\end{equation}
Here $v_c$ measures variation of the gain with the exterior matching, and $\theta_c$ is the probability that $P_c\subseteq M_1$ and $Q_c\subseteq M_2$. The factor two in $V(f)$ accounts for the reverse assignment. Unrestricted sums over $c$ count each geometric even cycle once.

Write $(M_1^c,M_2^c)$ for the pair obtained by a cycle exchange on a discrepancy cycle $c$ of $(M_1,M_2)$. The expected sum of squared gains satisfies
\begin{equation}\label{eq:cycle-square}
 2\Var_\pi(f)\le S(f):=
 \E_{\pi\otimes\pi}\sum_{c\text{ discrepancy cycle}}[f(M_1)-f(M_1^c)]^2
 =2\sum_c\theta_c\,\E_{\pi_c}d_c^2.
\end{equation}
For the inequality, condition on the multiset union $\mathcal U=M_1\uplus M_2$, counting common edges twice. Given $\mathcal U$, the assignments of the two alternating patterns on each discrepancy cycle to $M_1,M_2$ are independent fair bits, since every assignment preserves $\pi(M_1)\pi(M_2)$. The conditional laws of $f(M_1)$ and $f(M_2)$ agree. Bounding the squared difference by twice the sum of the centered squares therefore gives a factor of $4$. For each bit with the others fixed, the conditional variance of $f(M_1)$ is one quarter of its squared change. Variance tensorization cancels these factors, giving
\[
 \E[(f(M_1)-f(M_2))^2\mid\mathcal U]
 \le4\Var(f(M_1)\mid\mathcal U)
 \le\E\!\left[\sum_c(f(M_1)-f(M_1^c))^2\mid\mathcal U\right].
\]
Averaging gives the inequality in~\eqref{eq:cycle-square}.

For the equality, condition on $P_c\subseteq M_1,\ Q_c\subseteq M_2$, an event of probability $\theta_c$. The remainders $U=M_1\setminus P_c$ and $W=M_2\setminus Q_c$ are independent with law $\pi_c$, since the fixed pattern weights cancel upon conditioning. The reverse assignment contributes equally.

Let $q(M,M')$ be the two-edge switch proposal probability, including the factor $1/3$. Put $F(M_1,M_2)=f(M_1)+g(M_2)$ and $F^c=F(M_1^c,M_2^c)$. The energies of two-edge switches and cycle exchanges are respectively $L(f)$ and $C(F)$:
\begin{equation}\label{eq:pair-energy}
 \begin{aligned}
 \cE(F)&=L(f)+C(F),\\
 L(f)&=\sum_{\{M,M'\}\text{ two-edge switch}}
 q(M,M')\min\{\pi(M),\pi(M')\}[f(M)-f(M')]^2,\\
 C(F)&=\frac16\E_{\pi\otimes\pi}
 \left[\frac1k\sum_{c\text{ discrepancy cycle}}(F-F^c)^2\right].
 \end{aligned}
\end{equation}
Here $k$ is the number of discrepancy cycles, and the last integrand is zero when $k=0$.

Condition again on $P_c\subseteq M_1,\ Q_c\subseteq M_2$. Define $d_{g,c}$ by replacing $f$ with $g$ in~\eqref{eq:cycle-data}. Then $F-F^c=d_c(U)-d_{g,c}(W)$, and independence of $U,W\sim\pi_c$ gives
\[
 \E(F-F^c)^2
 =v_c+\Var_{\pi_c}(d_{g,c})
  +(\E_{\pi_c}d_c-\E_{\pi_c}d_{g,c})^2\ge v_c.
\]
The reverse assignment gives the same bound. Summing over cycles and using $k\le K$ yields
\begin{equation}\label{eq:exchange-control}
 V(f)\le6K C(F).
\end{equation}
By~\cref{eq:cycle-square,eq:pair-energy,eq:exchange-control}, it remains to bound $S(f)$ by a polynomial in $n,r$ times $L(f)+V(f)$. We do this by dissecting each discrepancy cycle and comparing gains on its smaller parts.

\subsection{A hierarchy and independent isolation}\label{sec:basins}
To bound $S(f)$, we first organize each discrepancy cycle into intervals that can be isolated with a uniform weight loss. Fix a source $d=(M_1,M_2,c)$ indexing a summand of $S(f)$, with mass $\rho_d=\pi(M_1)\pi(M_2)$. Choose a minimum-$B$ edge of $c$ as its closing edge and number the complementary path $0,\ldots,m-1$, where $m=|c|$. Resolve ties and direction by vertex labels. All intervals below follow this path. We may take $P_c=\{01,23,\ldots,(m-2)(m-1)\}$ in~\eqref{eq:cycle-data}.

To make the hierarchy follow the path, define symmetric weights on $V(c)$ by
\begin{equation*}
 \Bflat_{ij}=\min_{i\le t<j}B_{t,t+1},\qquad i<j.
\end{equation*}
These weights form an ultrametric similarity on $V(c)$ satisfying $\Bflat\le B|_{V(c)}$, with equality on the edges of $c$. Equality on the closing edge follows because a cycle cannot have a unique minimum edge weight under an ultrametric similarity. The components of the graph on $V(c)$ with edges $ij$ satisfying $\Bflat_{ij}\ge t$ are its \emph{threshold clusters}. These intervals form an inclusion tree with root $V(c)$, singleton leaves, and maximal proper subclusters as children. Each nonsingleton node $I$ has height $h(I)=\min_{e\in\binom I2}\Bflat_e$, where $\binom I2$ is the set of two-element subsets of $I$. In this tree, heights strictly increase towards nonsingleton children.

Form a binary tree $\mathcal T$ by merging the first two children, then their union with each remaining child in path order. The inserted unions, called \emph{intermediate nodes}, inherit the original parent's height and remain nested with or disjoint from threshold clusters. Write $p(I)$ for a nonroot node's parent.

For perfect matchings $P,Q$ on $V(c)$, let $N_I(P,Q)=|P\cap\binom I2|+|Q\cap\binom I2|$ be the total number of their edges inside $I$, with multiplicity. We have $\Bflat_{ij}=h(I)$ for the smallest node $I$ containing $i,j$. Successive log-height differences along its path from the root telescope to $\log\Bflat_{ij}-\log h(V(c))$. Summing over the $m$ edge occurrences by node gives, by~\eqref{eq:notation},
\[
 \log\bigl(w_{\Bflat}(P)w_{\Bflat}(Q)\bigr)
 =m\log h(V(c))+
 \sum_{\substack{I\in\mathcal T\\I\ne V(c),\ |I|\ge2}}
 N_I(P,Q)\log\frac{h(I)}{h(p(I))}.
\]
The coefficients are nonnegative and zero at intermediate nodes.

Closing an interval $E=[i,\ldots,j]$ of even size gives two alternating matchings on $E$:
\begin{equation*}
 T_E=\{i(i+1),\ldots,(j-1)j\},\qquad
 C_E=\{ij,(i+1)(i+2),\ldots,(j-2)(j-1)\}.
\end{equation*}
Starting at the root, the recursion selects even nodes of $\mathcal T$ of size at least four as \emph{basins}.

Isolate each basin directly from $(M_1,M_2)$, placing the matching containing $T_E$ first by interchanging the two whole matchings if necessary. For proper $E$, let $x=i-1,y=j+1$ on the original $c$-circle. Reconnect the second matching's two edges crossing $E$ by the two-edge switch
\begin{equation}\label{eq:isolation}
 xi+jy\longmapsto ij+xy.
\end{equation}
Edge-set sums denote multiset unions. For the root, do nothing. The resulting pair and the pair after a cycle exchange on $E$ are
\[
 Y_E=(T_E+R_E,C_E+S_E),\qquad
 Y_E^E=(C_E+R_E,T_E+S_E),
\]
where $R_E,S_E$ are perfect matchings on $[n]\setminus E$, which may differ.

The two-edge switch in~\eqref{eq:isolation} cannot decrease the number of edges inside any threshold cluster $I$. If $I\subseteq E$ or $I\cap E=\varnothing$, neither removed edge lies inside $I$. If $E\subseteq I$, the change is
$1+\1_I(x)\1_I(y)-\1_I(x)-\1_I(y)=(1-\1_I(x))(1-\1_I(y))\ge0$.
The tree and heights are unchanged, so the product formula above shows that this switch cannot decrease the product of the two matching weights on $V(c)$ under $\Bflat$.
Write $a_e=B_e\chi_e$, with $1\le\chi_e\le r$. The unchanged weights outside $V(c)$ cancel. Since $B=\Bflat$ on the discrepancy cycle and $B\ge\Bflat$ on added chords, the $B$-weight product cannot decrease. The removed factors satisfy $\chi_{xi}\chi_{jy}\le r^2$, while the new factors satisfy $\chi_{ij}\chi_{xy}\ge1$. Hence
\begin{equation}\label{eq:isolation-weight}
 \rho_d\le r^2\pi(Y_{E,1})\pi(Y_{E,2}),
\end{equation}
also after a cycle exchange on $E$. Every isolation starts from the original pair, so the factor $r^2$ does not accumulate with depth.

\subsection{Cell identities and centering}\label{sec:centering}
We bound $S(f)$ by cell identities, basin dissections, and centering to reconcile exterior matchings.

Fix a basin $E=[i,\ldots,j]$. Here $u\to v$ denotes the arc from $u$ to $v$ along the directed $E$-boundary $(i,i+1,\ldots,j,i)$. A quadrilateral with cyclically ordered corners $(a,b,c,d)$ is called a \emph{cell}. For the local identity, assume the corners have alternating index parities. The four intervening boundary paths are then \emph{odd arcs}, with odd numbers of edges and even numbers of vertices. Put $H=(a\to b,c\to d)=(Y,X)$ and $G=(b\to c,d\to a)$, and assume at least one $G$ arc is \emph{short}, meaning it has one edge. For an odd arc $J=(v_0,\ldots,v_s)$, define
\begin{equation*}
 T_J=\{v_0v_1,v_2v_3,\ldots,v_{s-1}v_s\},\quad
 N_J=\{v_1v_2,v_3v_4,\ldots,v_{s-2}v_{s-1}\},\quad
 C_J=N_J\cup\{v_0v_s\}.
\end{equation*}
$T_J,C_J$ match all arc vertices and $N_J$ only the interior. Short arcs have $T_J=C_J$ (\cref{fig:arc-patterns}).

\begin{figure}[htbp]
\centering
\begin{tikzpicture}[x=.43cm,y=.7cm,every node/.style={font=\footnotesize}]
\foreach \base/\lab in {0/{T_J},8/{N_J},16/{C_J}}{
 \node[font=\small] at ({\base+2.5},1.15) {$\lab$};
 \foreach \i in {0,...,5}{
  \coordinate (p\base-\i) at ({\base+\i},0);
  \fill (p\base-\i) circle (1.1pt);
  \node[below=3pt] at (p\base-\i) {$v_{\i}$};
 }
}
\foreach \a/\b in {0/1,2/3,4/5}{\draw[thick] (p0-\a) to[bend left=38] (p0-\b);}
\foreach \a/\b in {1/2,3/4}{\draw[thick] (p8-\a) to[bend left=38] (p8-\b);}
\foreach \a/\b in {1/2,3/4}{\draw[thick] (p16-\a) to[bend left=38] (p16-\b);}
\draw[thick] (p16-0) to[bend left=27] (p16-5);
\end{tikzpicture}
\caption{Patterns on an odd arc. $T_J,C_J$ match all its vertices, whereas $N_J$ matches only its interior.}\label{fig:arc-patterns}
\end{figure}
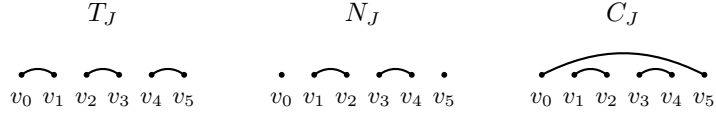

For any basin or odd arc $X$ and a matching $U$ on its complementary vertices, write $d_X(U)=f(T_X+U)-f(C_X+U)$. For the fixed basin $E$, take an exterior matching $R$ and set $O_0=T_E+R$, $O_1=C_E+R$. An arc $J$ is \emph{through} in $O_\epsilon$ when $T_J\subseteq O_\epsilon$. A pair is through when both arcs are. Since the four arcs are odd, their first edges alternate between $T_E$ and $C_E$. Thus $G$ is through in one of $O_0,O_1$ and $H$ in the other. For a through arc, put
$\Delta_J(R)=f(O_\epsilon)-f((O_\epsilon\setminus T_J)+C_J)$.

Set $N_G=N_{b\to c}\cup N_{d\to a}$. For $Z\subseteq\{Y,X\}$, put
\begin{equation}\label{eq:cell-patterns}
 M_Z=N_G+\sum_{J\in Z}C_J+\sum_{J\in\{Y,X\}\setminus Z}T_J.
\end{equation}
Each $M_Z+R$ is a perfect matching. Starting from $M_\varnothing+R$, in which $H$ is through, replace $T_Y$ by $C_Y$ and then $T_X$ by $C_X$ to obtain $M_{\{Y\}}+R$ and $M_{\{Y,X\}}+R$. On a nonshort arc, this is a cycle exchange along the arc and its endpoint chord.

Writing $N=N_G+N_Y+N_X$ for the matching on the four arc interiors, define
\[
 A_H=N+\{ab,cd\},\qquad A_G=N+\{bc,da\}.
\]
Here $A_H=M_{\{Y,X\}}$, and $A_G+R,A_H+R$ differ by a two-edge switch. Exchanging $Y$ changes the exterior seen by the subsequent exchange on $X$. The change in the $X$ gain is
\begin{equation}\label{eq:cell-error}
 E_H(R)=d_X((M_{\{Y\}}+R)\setminus T_X)
        -d_X((M_\varnothing+R)\setminus T_X).
\end{equation}
Thus exchanging both $H$ arcs has gain $\Delta_Y+\Delta_X+E_H$. The corresponding $G$ error is zero because $G$ has a short arc. Also $E_H=0$ if either $H$ arc is short. Put
$\ell_{\mathrm{cell}}(R)=f(A_G+R)-f(A_H+R)-E_H(R)$ when $G$ is through in $O_0$, and use its negative when $G$ is through in $O_1$. Comparing the exchanges on $G$ and $H$ gives the signed cell identity of~\cite[Lemma 6.3]{OpenAI}
\begin{equation}\label{eq:one-cell}
 d_E(R)=\ell_{\mathrm{cell}}(R)
       +\sum_{J\text{ through in }O_0}\Delta_J(R)
       -\sum_{J\text{ through in }O_1}\Delta_J(R).
\end{equation}
The sums range over the four arcs of this cell.

We now use $\mathcal T$ to select such cells, leaving even subpolygons as \emph{polygon holes} for recursion. A basin may yield several cells. The children of $E$ have the same parity. If both are even, say $[a,\ldots,b]$ and $[c,\ldots,d]$, select cell $(a,b,c,d)$, leaving both children as polygon holes. For odd children $I_1,I_2$, split $E$ along the chord joining their last vertices into $(\operatorname{last}(I_2),I_1)$ and $(\operatorname{last}(I_1),I_2)$. For an odd interval $O=[p,\ldots,q]$, $(a,O)$ denotes the polygon with cyclic vertex order $(a,p,\ldots,q)$, where $a$ immediately precedes $p$ on the $E$-boundary. A singleton $O$ gives an edge. For $|O|\ge3$, the children of $O$ in $\mathcal T$ split $(a,O)$ into a cell, a polygon hole on the even child, and a smaller even polygon on the odd child plus one vertex:
\begin{enumerate}
\item $[p,\ldots,s]$ odd, $[c,\ldots,q]$ even, $c=s+1$: select $(a,s,c,q)$ and continue with $(a,[p,\ldots,s])$.
\item $[p,\ldots,b]$ even, $[u,\ldots,q]$ odd, $u=b+1$: select $(q,a,p,b)$ and continue with $(b,[u,\ldots,q])$.
\end{enumerate}
Every new chord closes an even interval or joins $a$ to the last vertex of $O$ in a polygon $(a,O)$. Its endpoints therefore have opposite index parities. Thus each cell has alternating corner parities, and $G$ has a short arc in every case, as required for~\eqref{eq:one-cell}. The ordering places $X$ on the even child for the weight comparisons in~\cref{sec:guides}.

For each polygon hole $h=[u,\ldots,v]$ of $E$, put $s_{Eh}=(-1)^{u-i}$. The matching containing $T_h$ and its exterior are
\begin{equation}\label{eq:hole-context}
 M_{Eh}(R)=\begin{cases}T_E+R,&s_{Eh}=+1,\\C_E+R,&s_{Eh}=-1,
 \end{cases}\qquad
 U_{Eh}(R)=M_{Eh}(R)\setminus T_h.
\end{equation}
To derive the basin identity, sum~\eqref{eq:one-cell} over the $k_E$ cells selected for $E$ before recursion. Each new cell attaches along one shared diagonal without crossings, so their adjacency graph is a tree. Each shared diagonal $uv$ determines complementary $E$-boundary arcs $J,\bar J$ whose exchanges both reach
$R+N_J+N_{\bar J}+\{uv\}$, so their signed gains sum to $d_E(R)$. These $k_E-1$ diagonals cancel all but one copy of $d_E(R)$. Original boundary edges have zero gain, yielding
\begin{equation}\label{eq:partial-identity}
 d_E(R)=\sum_{\mathrm{cells}}\ell_{\mathrm{cell}}(R)
       +\sum_{h\text{ polygon hole}}s_{Eh}d_h(U_{Eh}(R)),
\end{equation}
where both sums range over the dissection of $E$.
Polygon holes of size two have $T_h=C_h$ and zero gain. Starting at the root, recurse on larger polygon holes as child basins, isolating each directly by~\eqref{eq:isolation}. They are strict descendants in $\mathcal T$ with disjoint sibling subtrees, so no basin repeats. The inherited exterior $U_{Eh}(R_E)$ may differ from the independently isolated exterior $R_h$. Following~\cite[Section 6.4]{OpenAI}, let $\pi_X$ be the matching law outside $X$, and set $m_X=\E_{\pi_X}d_X$, $z_X=d_X-m_X$, and $v_X=\Var_{\pi_X}(d_X)$. For basins and nonshort arcs $X$, $T_X\cup C_X$ is a geometric cycle $c$, and $v_X=v_c$ for either alternating pattern. Short arcs satisfy $d_X=m_X=v_X=0$.

Set $D_E=d_E(R_E)$ for each independently isolated basin. The two gains on a child basin $h$ use the same $T_h,C_h$, so subtracting the same mean $m_h$ gives
\[
 d_h(U_{Eh}(R_E))=D_h+
 \bigl[z_h(U_{Eh}(R_E))-z_h(R_h)\bigr].
\]
The $D_h$ terms will cancel between parent and child identities. The exteriors may differ outside $E$ because isolation can interchange the source matchings. With $h$ ranging over the child basins of $E$, substitution into~\eqref{eq:partial-identity} yields
\begin{equation}\label{eq:basin-centered}
 D_E-\sum_h s_{Eh}D_h
 =\sum_{\mathrm{cells}}\ell_{\mathrm{cell}}(R_E)
 +\sum_h s_{Eh}\bigl[z_h(U_{Eh}(R_E))-z_h(R_h)\bigr]
 =:\mathcal A_E.
\end{equation}
Replacing $d_X$ by $z_X$ in~\eqref{eq:cell-error} likewise leaves $E_H$ unchanged.

Set $\tau_c=1$ at the root and $\tau_h=\tau_Es_{Eh}$ for each child basin. Multiplying~\eqref{eq:basin-centered} by $\tau_E$ and summing over basins cancels every nonroot gain:
\begin{equation}\label{eq:root-expansion}
 D_c=\sum_E\tau_E\mathcal A_E.
\end{equation}
At the root, $D_c$ is evaluated after interchanging the source matchings if needed to place $P_c$ first. Conditional on either assignment of $P_c,Q_c$ to $M_1,M_2$, this first exterior has law $\pi_c$. Both assignments have probability $\theta_c$, giving
\begin{equation}\label{eq:root-square}
 \sum_{d=(M_1,M_2,c)}\rho_dD_c^2
 =2\sum_c\theta_c\,\E_{\pi_c}d_c^2=S(f).
\end{equation}
To bound the number of gain terms, note that a full dissection of $E$ has $|E|/2-1$ cells, giving
\begin{equation}\label{eq:cell-count}
 k_E=\frac{|E|}{2}-1-\sum_{h\text{ polygon hole}}\left(\frac{|h|}{2}-1\right).
\end{equation}
Summing~\eqref{eq:cell-count} cancels the polygon-hole terms. With $N_b$ basins and $F_b=\sum_{E\text{ basin}}k_E$ cells,
\begin{equation}\label{eq:basin-count}
 N_b\le n,\qquad F_b=|c|/2-1<n/2,\qquad
 \#\{\text{parent--child basin links}\}=N_b-1.
\end{equation}
The expansion contains only two-edge switch gains and centered evaluations. We next bound their source masses to control $S(f)$. \Cref{fig:basins} compares the two gains on a child basin.

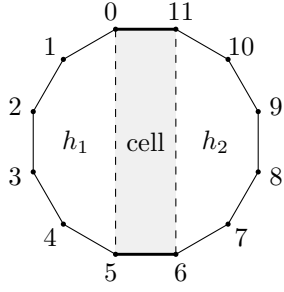
\begin{figure}[htbp]
\centering
\begin{minipage}[c]{0.48\textwidth}
\centering
\begin{tikzpicture}[scale=.88,every node/.style={font=\small}]
\foreach \i in {0,...,11}{\coordinate (v\i) at ({105+30*\i}:1.75);}
\fill[black!6] (v0)--(v5)--(v6)--(v11)--cycle;
\draw (v0)--(v1)--(v2)--(v3)--(v4)--(v5)--(v6)--(v7)--(v8)--(v9)--(v10)--(v11)--cycle;
\draw[dashed] (v0)--(v5) (v6)--(v11);
\draw[line width=1.1pt] (v5)--(v6) (v11)--(v0);
\foreach \i in {0,...,11}{\fill (v\i) circle (1.1pt);\node at ({105+30*\i}:2.03) {$\i$};}
\node at (-1.05,0) {$h_1$};\node at (1.05,0) {$h_2$};\node at (0,0) {cell};
\end{tikzpicture}
\end{minipage}
\begin{minipage}[c]{0.49\textwidth}
\centering
\small
$h_1=[0,\ldots,5],\quad h_2=[6,\ldots,11]$
\par\smallskip
\begin{tabular}{@{}ll@{}}
\toprule
Source & Gain on $h_1$ \\
\midrule
Inherited from $E$ & $d_{h_1}(T_{h_2}+R_E)$ \\
Independent isolation & $d_{h_1}(R_{h_1})$ \\
\bottomrule
\end{tabular}
\end{minipage}
\caption{One cell ($k_E=1$) and two polygon holes in $E=[0,\ldots,11]$. Here $U_{Eh_1}(R_E)=T_{h_2}+R_E$. The first gain minus the second is the centered error in $\mathcal A_E$.}\label{fig:basins}
\end{figure}

\subsection{Weight comparisons and congestion}\label{sec:guides}
We bound the squared contributions in~\eqref{eq:root-expansion} by comparing weights and bounding how many terms share an output. For each source $d$, assign each term an ordered pair of perfect matchings $(M,Q)$. The first matching determines the function evaluation, while the \emph{guide} $Q$ is used to compare weights. For centered terms, we retain $T_X\subseteq M$ and $C_X\subseteq Q$.

\paragraph{Cell comparisons.}
For an odd arc $J=(v_0,\ldots,v_s)$ of the current cell, obtain $V_J$ from $T_J$ by replacing $v_0v_1,v_{s-1}v_s$ with $v_1v_{s-1}$ when $s\ge3$. Set $V_J=\varnothing$ when $s=1$. Thus $V_J$ matches only the arc interior. Set $V_G=V_{b\to c}\cup V_{d\to a}$ and, for $Z\subseteq\{Y,X\}$, define
\[
 Q_Z=V_G+\sum_{J\in Z}T_J+\sum_{J\in\{Y,X\}\setminus Z}C_J.
\]
Together with $M_Z$ in~\eqref{eq:cell-patterns}, these give full matching pairs $(M_Z+R_E,Q_Z+S_E)$. Use $Q_{\{Y,X\}}+S_E$ for the switch between $A_G+R_E,A_H+R_E$, and $Q_{\{Y\}}+S_E,Q_\varnothing+S_E$ for the respective terms of~\eqref{eq:cell-error}. For $Z=\varnothing,\{Y\}$, the guide contains $C_X$. Before guide repair, the multiset union is independent of $Z$. If the basin's children are even, the $G$ chords cross the children and the $H$ chords lie within them. The comparison using threshold clusters in~\eqref{eq:isolation} shows that the local $\Bflat$-weight products do not decrease from the isolated pair.

For two odd children, modify only the guide by two-edge switches outside $X$, preserving $C_X$. Write $b=\Bflat$ and $b(P)=\prod_{e\in P}b_e$. We compare the isolated product $b(T_E)b(C_E)$ with the local output products, since exterior factors cancel. Successors and predecessors refer to the whole $E$-boundary.

For an odd prefix $[p,\ldots,s]$ and even suffix $X=[c,\ldots,d]$, the cell is $(a,s,c,d)$. If $d\to a$ is short, no repair is needed. For a nonshort $d\to a$, put $x=d+1$, $y=a-1$, $p=a+1$, and
$\gamma=b_{dx}b_{ya}/(b_{da}b_{xy})$.
Cancellation gives $b(T_E)b(C_E)/(b(A_G)b(Q_{\{Y,X\}}))=\gamma$, and the ratio at $A_H$ is $\gamma b_{sc}/b_{cd}\le\gamma$. Apply
\begin{equation*}
 Q_{\{Y,X\}},Q_{\{Y\}}:\ xy+ap\longmapsto ay+px,\qquad
 Q_\varnothing:\ xy+as\longmapsto ay+sx.
\end{equation*}
Vertices $a,x$ lie outside $[p,\ldots,d]$. For each, the $b$-weight is constant on this interval, so each switch multiplies the guide's weight by $\gamma$. Both switches leave $X$ untouched.

For an even prefix $X=[p,\ldots,b]$ and odd suffix $[u,\ldots,d]$, the cell is $(d,a,p,b)$, with $u=b+1$. If $b\to d$ is short, the $A_G,A_H$ ratios are $1,b_{bu}/b_{pb}\le1$. Otherwise set $v=d-1$, $x=d+1$, giving
\begin{equation*}
 \frac{b(T_E)b(C_E)}{b(A_G)b(Q_{\{Y,X\}})}
 =\frac{b_{bu}b_{vd}}{b_{bd}b_{uv}}
 =\gamma:=\frac{b_{vd}}{b_{uv}},\qquad
 \frac{b(T_E)b(C_E)}{b(A_H)b(Q_{\{Y,X\}})}
 =\gamma\frac{b_{bu}}{b_{pb}}\le\gamma.
\end{equation*}
Here $b_{bu}=b_{bd}$, $b_{ap}=b_{ad}$, and $b_{pb}\ge b_{bu}$ by the interval hierarchy. Apply
\begin{equation*}
 Q_{\{Y,X\}},Q_{\{Y\}}:\ uv+dx\longmapsto ux+vd,\qquad
 Q_\varnothing:\ uv+da\longmapsto ua+vd.
\end{equation*}
Vertices $x,a$ lie outside $[p,\ldots,d]$. For each, the $b$-weight is constant on this interval. Each switch multiplies the guide's weight by $\gamma$ and leaves $X$ untouched. If $Y$ is short, $E_H=0$, and the switch in $Q_{\{Y,X\}}$ remains valid with $x=a$. All these switches use four distinct vertices.

To pass to $a$-weights, we bound the number of changed edge occurrences, with multiplicity. In the following sums, $J$ ranges over the four cell arcs and $\operatorname{chord}_J$ denotes its endpoint chord. The isolated pair has union $\sum_J(T_J+N_J)$ on $E$. Before guide repair, $M_Z+Q_Z$ equals
\[
 \sum_JN_J+\sum_{J\in\{Y,X\}}(\operatorname{chord}_J+T_J)+V_G.
\]
This adds the two $H$ chords and replaces $T_J$ by $V_J$ on each $G$ arc. A short arc removes one occurrence, while a nonshort one removes two and adds one. Since $G$ has a short arc, at most three occurrences are removed and three added. Isolation and guide repair each remove at most two further occurrences, giving $2+3+2=7$. Replacing $A_H$ by $A_G$ raises this bound to nine. Additions satisfy the same bounds because each pair contains $n$ edge occurrences. After guide repair, let $M,M'$ be the two switch matchings and $Q_{\mathrm{loc}}$ their common full guide. For either term of~\eqref{eq:cell-error}, let $(M,Q)$ be its full pair. As in~\eqref{eq:isolation-weight}, compare with $B$ using $1\le a_e/B_e\le r$ to obtain
\begin{equation}\label{eq:cell-comparisons}
 \rho_d\le r^9\min\{\pi(M),\pi(M')\}\pi(Q_{\mathrm{loc}}),\qquad \rho_d\le r^7\pi(M)\pi(Q),\qquad C_X\subseteq Q.
\end{equation}

\paragraph{Polygon-hole comparisons.}
For a nontrivial polygon hole $h=[u,\ldots,v]$ of $E$, the two terms in~\eqref{eq:basin-centered} are $z_h(U_{Eh}(R_E))$ and $z_h(R_h)$. For the first, take $M=M_{Eh}(R_E)\supseteq T_h$. By~\eqref{eq:hole-context}, this is the first matching of $Y_E$, after a cycle exchange on $E$ if $s_{Eh}=-1$. Let $x,y$ immediately precede and follow $h$ on the $E$-boundary. In the second matching, apply the two-edge switch
\begin{equation*}
 xu+vy\longmapsto uv+xy.
\end{equation*}
The resulting guide $Q$ contains $C_h$. Since $h$ is a node of $\mathcal T$, the threshold-cluster argument in~\eqref{eq:isolation} shows that this switch cannot decrease the $\Bflat$-weight product. Both this switch and the isolation of $E$ use $\Bflat$ on $V(c)$, even if the switch removes $E$'s closing chord. Together they remove at most four edge occurrences from the source, giving
\begin{equation}\label{eq:hole-comparison}
 \rho_d\le r^4\pi(M)\pi(Q),\qquad T_h\subseteq M,\quad C_h\subseteq Q.
\end{equation}
For $z_h(R_h)$, use the independently isolated pair $Y_h=(T_h+R_h,C_h+S_h)$. By~\eqref{eq:isolation-weight},
\begin{equation}\label{eq:child-comparison}
 \rho_d\le r^2\pi(T_h+R_h)\pi(C_h+S_h).
\end{equation}
Taking $X=h$ here, these bounds and~\eqref{eq:cell-comparisons} give $\rho_d\le r^7\pi(M)\pi(Q)$ for every centered evaluation, with $T_X\subseteq M$ and $C_X\subseteq Q$.
For a two-edge switch, choose its endpoint of smaller $\pi$-probability as $M$, breaking ties by a fixed matching order. Its guide is common to both endpoints.

\paragraph{One encoding for all levels.}\label{sec:accounting}
An occurrence is a source $d=(M_1,M_2,c)$ and a term of~\eqref{eq:root-expansion}, fully expanded before repeated terms are combined. There are five types: a two-edge switch gain, either evaluation of a cell error, or the parent or child evaluation of a polygon-hole error. Adapting~\cite[Lemma 6.4]{OpenAI}, the following record and output pair recover each occurrence.

Record seven vertex slots: the least vertex of $c$, the basin endpoints, and the four cell corners. For a parent--child link, replace the corners by child endpoints and fixed padding. Record the type and four binary flags: which alternating pattern on $c$ belongs to $M_1$, parity of interchanges of the two whole matchings, choice of switch endpoint, and order of the $H$ arcs. Fix unused flags. Two sorted lists, padded to nine entries each, record deletions and additions in the multiset union, retaining multiplicities but omitting matching labels. Cycle exchanges preserve this union and require no entries.

Recover the occurrence from $(M,Q)$ and this record as follows.
\begin{enumerate}
\item Reverse the recorded additions and deletions in $M\uplus Q$ to recover $M_1\uplus M_2$. Its discrepancy cycle containing the first recorded vertex is $c$.
\item On $c$, the flag for $M_1$ fixes both alternating patterns. All edited edges lie in $c$, so outside $c$ we only undo the recorded interchange of the two whole matchings.
\item Reconstruct the hierarchy and selected basins from $c$, $B$, and the fixed tie rules.
\item Locate the term using the basin endpoints, cell corners or child endpoints, type, and flags. Each child basin has a unique parent in the recursion.
\end{enumerate}
The seven vertex slots have at most $n^7$ choices, the type and flags have $5\cdot2^4=80$, and each list entry has at most $1+n^2$. Thus each output pair has at most
\begin{equation*}
 80n^7(1+n^2)^{18}
 <(10n)^{60}=:\enc
\end{equation*}
occurrences in total, over all sources, basins, and depths.

\paragraph{Summing the loads.}
The load is the sum of $\rho_d$ over occurrences assigned to a fixed switch or centered evaluation. For a fixed unordered switch $\{M,M'\}$, use the endpoint $M$ chosen above. By~\eqref{eq:cell-comparisons}, at most $\enc$ occurrences per $(M,Q)$ and $\sum_Q\pi(Q)\le1$ give a load of at most $\enc r^9\pi(M)$. To compare with $L(f)$, use $q(M,M')\ge q_*:=1/(30n^2)$ to bound the load by
\begin{equation}\label{eq:local-load}
 \frac{\enc r^9}{q_*}\,
 q(M,M')\min\{\pi(M),\pi(M')\}.
\end{equation}
For a nonshort directed odd path $X$ with its closing chord and a matching $M\supseteq T_X$, let $H(X,M)$ be the load of $z_X(M\setminus T_X)$. Summing~\cref{eq:cell-comparisons,eq:hole-comparison,eq:child-comparison} over guides $Q\supseteq C_X$, with at most $\enc$ occurrences per $(M,Q)$, gives
\begin{equation*}
 H(X,M)\le\enc r^7\pi(M)p_X,\qquad
 p_X=\pi(C_X\subseteq Q).
\end{equation*}
Multiply by $z_X(M\setminus T_X)^2$ and extend the nonnegative sum to all $M\supseteq T_X$. Under $\pi$ conditioned on $T_X\subseteq M$, the exterior has law $\pi_X$, so the total centered contribution is at most
\begin{equation}\label{eq:centered-sum}
 \enc r^7\sum_Xp_X\sum_{M\supseteq T_X}\pi(M)z_X(M\setminus T_X)^2=\enc r^7\sum_Xt_Xp_Xv_X,\qquad t_X=\pi(T_X\subseteq M).
\end{equation}
The guide and matching sums supply $p_X$ and $t_X$, respectively. Closing $X$ gives a geometric cycle $c$ with $t_Xp_X=\theta_c$ and $v_X=v_c$. The closing edge and direction give at most $2|c|\le2n$ representations of $c$. Since $V(f)=2\sum_c\theta_cv_c$,
\begin{equation}\label{eq:physical-sum}
 \sum_Xt_Xp_Xv_X\le nV(f).
\end{equation}

Each cell contributes one switch gain and at most two centered evaluations, and each nontrivial parent--child link contributes two centered evaluations. By~\eqref{eq:basin-count}, the sum $W$ of absolute coefficients in the expansion of~\eqref{eq:root-expansion} satisfies
\begin{equation*}
 W\le3F_b+2(N_b-1)<4n.
\end{equation*}
Apply Cauchy--Schwarz once to this expansion, multiply by $\rho_d$, and sum over sources. \Cref{eq:root-square,eq:local-load,eq:centered-sum,eq:physical-sum} give
\begin{equation*}
 S(f)\le4n\enc\left[\frac{r^9}{q_*}L(f)+nr^7V(f)\right].
\end{equation*}
Using~\cref{eq:cycle-square,eq:exchange-control}, $q_*^{-1}=30n^2$, $K\le n/4$, and $r\ge1$, we obtain
\[
 \Var_\pi(f)
 \le60n^3\enc r^9L(f)+12n^2K\enc r^7C(F)
 \le60n^3\enc r^9\cE(F)
 \le\Cp\cE(F).
\]
This proves \cref{thm:pair}.\hfill$\square$

\section{A replicated sampler}\label{sec:sampling}
\Cref{thm:pair} controls additive functions of two matchings. Sampling requires a variance bound for arbitrary functions of the chain's state. We connect the uniform and target laws through nearby levels, then apply OpenAI's replication argument with several matchings per level.

\subsection{A path of nearby laws}
Suppose integers $D,Q\ge1$ are given and
\begin{equation}\label{eq:structural-bounds}
 B_0\le a\le(D+1)B_0,\qquad D^{-1}\le a_e\le2\cdot4^Q
\end{equation}
for some positive ultrametric similarity $B_0$. Multiplying every activity by $D$ multiplies every matching weight by $D^{n/2}$, so $\pi_{Da}=\pi_a$. Set $r=D+1$, $c=1+1/n$, and take the least integer $T$ such that $c^T\ge2D4^Q$. Define
\begin{equation}\label{eq:capped-path}
 a_e^{(t)}=\min\{c^t,Da_e\},\qquad
 B_e^{(t)}=\min\{c^t,DB_{0,e}\},\qquad
 \pi_t=\pi_{a^{(t)}}\quad(0\le t\le T).
\end{equation}
Since $1\le Da_e\le2D4^Q$, $\pi_0$ is uniform and $\pi_T=\pi_a$. Capping preserves the ultrametric inequality for $B^{(t)}$, and $\min\{u,rx\}\le r\min\{u,x\}$ gives
\begin{equation*}
 B^{(t)}\le a^{(t)}\le rB^{(t)},
\end{equation*}
so \cref{thm:pair} applies at every level with the same $\Cp$.

Between adjacent levels, edge weights increase by factors in $[1,c]$, so both matching weights and the partition function increase by factors in $[1,c^{n/2}]$. Since $c^{n/2}<2$ and $\log c\ge1/(2n)$,
\begin{equation}\label{eq:adjacent-laws}
 \frac12\le\frac{\pi_{t-1}(M)}{\pi_t(M)}\le2,\qquad
 T\le1+2n\log(2D4^Q)
 \le1+4n\log\max(D,2\cdot4^Q).
\end{equation}
The path uses only $a,D,Q$, with $B^{(t)}$ certifying the comparison. To sample $\pi_0$, repeatedly match the least unused vertex to a uniformly chosen unused partner. Each matching has probability $[(n-1)(n-3)\cdots1]^{-1}$.

\subsection{The replicated chain and its gap}
Use $q=32\Cp$ matchings $M_{t,j}$ per level, $0\le t\le T$, $1\le j\le q$, with stationary law $\nu=\bigotimes_{t=0}^T\pi_t^{\otimes q}$.
A \emph{coordinate} is one matching position $(t,j)$. This product law makes the matchings independent under $\nu$, with single-coordinate distribution, or \emph{marginal}, $\pi_t$. Number the positions from highest to lowest level, then by $j$. Thus lower-level positions come later in this order.

The level-zero positions are the \emph{base coordinates}. A \emph{refresh} independently resamples one such matching from $\pi_0$, leaving all others fixed. The \emph{guides} of a level-$t$ position, $t\ge1$, are the $q$ positions at level $t-1$. Pair updates use the two-edge switch, cycle exchange, and holding proposals of $\mathsf K$, with the higher-level matching first. All proposals, including cycle exchanges, use Metropolis acceptance for $\pi_t\otimes\pi_{t-1}$. The chain $P$ holds with probability $1/2$ and otherwise chooses uniformly among these $J=Tq^2+q$ updates. All updates, hence $P$, are $\nu$-reversible.

Use the energy $\cE_{\sigma,P}$ defined in \cref{sec:comparison}. For $\cE_i(f)$ (a base refresh) and $\cE_{xy}(f)$ (a pair update), fix the other coordinates, compute the local energy of $f$, then average over their values under $\nu$. For a refresh, the local energy equals the variance under its marginal.

The following lemma, extracted from OpenAI's replication proof~\cite[Lemma 7.2]{OpenAI}, converts an additive pair bound into a bound for every function of the full state.

\begin{lemma}\label{lem:replication}
Let $\xi=(\xi_1,\ldots,\xi_m)\sim\nu=\bigotimes_{x=1}^m\mu_x$, where $\mu_x$ is a law on a finite set $\mathcal X_x$. Let $I_0\subseteq[m]$ index base coordinates, and choose guide sets
\[
 \Gamma_x\subseteq\{x+1,\ldots,m\},\qquad |\Gamma_x|=q\ge1
 \quad(x\notin I_0).
\]
Refresh each $i\in I_0$ from $\mu_i$. For $x\notin I_0$, $y\in\Gamma_x$, let $K_{xy}$ be a $\mu_x\otimes\mu_y$-reversible kernel depending only on $(\xi_x,\xi_y)$, with energy $\cE_{xy}$. If $C>0$ and, for each such pair and all $a:\mathcal X_x\to\mathbb R$, $b:\mathcal X_y\to\mathbb R$,
\begin{equation}\label{eq:replication-assumption}
 \Var_{\mu_x}(a)\le C\cE_{xy}(a(\xi_x)+b(\xi_y)),
\end{equation}
then every $f:\prod_{x=1}^m\mathcal X_x\to\mathbb R$ satisfies
\begin{equation}\label{eq:replication-bound}
 \left(1-\frac{4C}{q}\right)\Var_\nu(f)
 \le\frac{2C}{q}\sum_{x\notin I_0}\sum_{y\in\Gamma_x}\cE_{xy}(f)
       +\sum_{i\in I_0}\cE_i(f).
\end{equation}
\end{lemma}
\begin{proof}
All norms are in $L^2(\nu)$. With $U_x=(\xi_1,\ldots,\xi_{x-1})$, revealing $\xi_x$ adds the increment $h_x=\E[f\mid U_x,\xi_x]-\E[f\mid U_x]$. By orthogonality, $\sum_x\norm{h_x}^2=\Var_\nu(f)$. For $x<y$, subtracting single-coordinate terms from the conditional mean $g_{xy}=\E[f\mid U_x,\xi_x,\xi_y]$ leaves the interaction part
\[
 r_{xy}=g_{xy}-\E[f\mid U_x,\xi_x]-\E[f\mid U_x,\xi_y]+\E[f\mid U_x].
\]
For a chosen guide pair, conditional on $U_x$, $g_{xy}-r_{xy}$ is additive in $\xi_x,\xi_y$, and its $\xi_x$ summand has averaged variance $\norm{h_x}^2$. Jensen's inequality gives $\cE_{xy}(g_{xy})\le\cE_{xy}(f)$. Using~\eqref{eq:replication-assumption}, $\cE(u-v)\le2\cE(u)+2\cE(v)$, and $\cE(v)\le2\norm v^2$ for stationary kernels yields
\begin{equation}\label{eq:doob-first}
 \norm{h_x}^2\le2C\cE_{xy}(f)+4C\norm{r_{xy}}^2.
\end{equation}
For fixed $y$, independence and the tower property give
\[
 r_{xy}=\E[h_y\mid U_x,\xi_x,\xi_y]-\E[h_y\mid U_x,\xi_y].
\]
Reveal $\xi_y$ first, then $\xi_1,\ldots,\xi_{y-1}$. These are orthogonal increments of $h_y$, so $\sum_{x<y}\norm{r_{xy}}^2\le\norm{h_y}^2$. Summing over guide pairs gives $\sum_{x\notin I_0}\sum_{y\in\Gamma_x}\norm{r_{xy}}^2\le\Var_\nu(f)$. At a base coordinate, refresh and Jensen's inequality give $\norm{h_i}^2\le\cE_i(f)$. Average~\eqref{eq:doob-first} over $\Gamma_x$ and sum over $x\notin I_0$:
\[
 \Var_\nu(f)\le\frac{2C}{q}\sum_{x\notin I_0}\sum_{y\in\Gamma_x}\cE_{xy}(f)
 +\frac{4C}{q}\Var_\nu(f)+\sum_{i\in I_0}\cE_i(f).
\]
Moving the variance term to the left proves~\eqref{eq:replication-bound}.
\end{proof}

Apply the lemma with $I_0$ the level-zero positions and $\Gamma_x$ the positions one level below $x$. By~\eqref{eq:adjacent-laws}, $\pi_t\otimes\pi_{t-1}\ge\tfrac12\pi_t\otimes\pi_t$ pointwise. The symmetric proposal rule is unchanged. For each unordered state pair, its Metropolis energy coefficient is the proposal probability times the smaller endpoint density, hence at least half its equal-level value. Thus \cref{thm:pair} gives~\eqref{eq:replication-assumption} with $C=2\Cp$. Since $q=16C$ and $P$ chooses each update with probability $1/(2J)$, \cref{lem:replication} gives
\begin{equation*}
 \Var_\nu(f)\le\frac16\sum_{x\notin I_0}\sum_{y\in\Gamma_x}\cE_{xy}(f)
 +\frac43\sum_{i\in I_0}\cE_i(f)
 \le\frac{8J}{3}\cE_{\nu,P}(f).
\end{equation*}
Hence $P$ has gap at least $3/(8J)$ and is irreducible: a nonconstant closed-class indicator would violate this bound. The holding probability makes $P$ lazy.

Adapting OpenAI's mixing and finite-precision analysis~\cite[Theorem 7.3 and Section 9]{OpenAI} to this gap bound gives the following sampler. Write $\bit(a)$ for the activities' total binary encoding length and $\|\mu-\nu\|_{\TV}=\tfrac12\sum_x|\mu(x)-\nu(x)|$ for total variation distance.

\begin{theorem}\label{thm:structural-sampler}
Given even $n\ge4$, positive rational $a$ satisfying~\eqref{eq:structural-bounds}, and an integer $b\ge1$, there is a sampler within $2^{-b}$ of $\pi_a$ in total variation. Its bit running time and fair-random-bit budget are polynomial in $n,D,Q,\bit(a),b$ on every execution. Only $a,D,Q,b$ are inputs.
\end{theorem}
\begin{proof}
The implementation in \cref{app:simulation} runs $P$ from fixed matchings and outputs one top-level matching, whose stationary law is $\pi_a$. It uses fixed time and random-bit budgets and bounds mixing and total rounding errors by $2^{-b-1}$ each.
\end{proof}

\section{Completion and marking}\label{sec:completion}
Completion gives the comparison required by \cref{thm:structural-sampler}. We bound its increase in total weight and mark edges by their weight contributions.

For positive activities $\lambda$, set $D=10^8(n+1)^4$ and define
\begin{equation}\label{eq:bottleneck}
 B_{ij}=\max\left(1,\max_{P:i\leadsto j}\min_{e\in P}\lambda_e\right),
 \qquad \mu=\lambda+B/D,
\end{equation}
where $P$ ranges over simple $i$-to-$j$ paths in $K_n$. Concatenating paths and removing loops gives the ultrametric inequality for $B$, and one-edge paths give $\lambda\le B$. To compute $B$, initialize $B_{ij}=\max(1,\lambda_{ij})$, $B_{ii}=\infty$, and, for each $k$, replace $B_{ij}$ by $\max(B_{ij},\min(B_{ik},B_{kj}))$. After step $k$, the allowed internal vertices are $1,\ldots,k$. This takes $O(n^3)$ comparisons, and every finite value is one or an input activity. The completion satisfies
\begin{equation*}
 B/D\le\mu\le(D+1)B/D,
\end{equation*}
providing the ultrametric comparison in~\eqref{eq:structural-bounds}.

\subsection{The increase in total weight}
For positive activities $a$, let $Z_a(-U)$ be the total weight of matchings leaving exactly the \emph{hole set} $U\subseteq[n]$ unmatched. The ratio $g_a(U)=Z_a(-U)/Z_a$ can exceed one and is not a probability. We omit braces, writing $g_a(ij)=g_a(\{i,j\})$, and have $g_a(\varnothing)=1$, $Z_a(-[n])=1$. Set $r_i(a)=\max_{j\ne i}g_a(ij)$. We call $\lambda$ \emph{balanced} if
\begin{equation*}
 \tfrac14\le r_i(\lambda)\le1\qquad\text{for every }i.
\end{equation*}
Thus all $g_\lambda(ij)\le1$, and for each $i$ some $j\ne i$ has $g_\lambda(ij)\ge1/4$. Growth of the total weight involves two-hole ratios, whose changes involve four-hole ratios. The next lemma bounds both.

\begin{lemma}[OpenAI~{\cite[Lemmas 3.2--3.3]{OpenAI}}]\label{lem:holes}
For positive activities $a$ and distinct $i,j,u,v$,
\begin{equation}\label{eq:fourholes}
 g_a(ijuv)\le g_a(ij)g_a(uv)+g_a(iu)g_a(jv)+g_a(iv)g_a(ju).
\end{equation}
If also $a\ge\lambda$ and $1/8\le r_i(a)\le4$ for every $i$, then for $|U|\in\{2,4\}$ and distinct $i,j\in U$,
\begin{equation}\label{eq:strongpath}
 g_a(U)\le10^5n/B_{ij}.
\end{equation}
\end{lemma}
\begin{proof}
For even hole sets $U_1,U_2$ and $x\in U_1\triangle U_2$, we first show
\begin{equation}\label{eq:hole-exchange}
 g_a(U_1)g_a(U_2)\le
 \sum_{y\in(U_1\triangle U_2)\setminus\{x\}}
 g_a(U_1\triangle\{x,y\})g_a(U_2\triangle\{x,y\}).
\end{equation}
Overlay an ordered pair of matchings missing $U_1,U_2$. The alternating path from $x$ ends at a unique $y\in(U_1\triangle U_2)\setminus\{x\}$. Exchanging their edges along it preserves the product of their weights and gives the displayed hole sets. For fixed $y$, the output path from $x$ recovers the input. Summing over endpoints and dividing by $Z_a^2$ proves~\eqref{eq:hole-exchange}. Taking $U_1=\{i,j,u,v\}$, $U_2=\varnothing$ gives~\eqref{eq:fourholes}, hence four-hole ratios are at most $3\cdot4^2=48$ when every $r_i(a)\le4$.

We also use edge insertion: for an even hole set $C$ and distinct $x,z\notin C$,
\begin{equation}\label{eq:edge-insertion}
 a_{xz}g_a(C\cup\{x,z\})\le g_a(C),
\end{equation}
because adjoining $xz$ injects matchings missing $C\cup\{x,z\}$ into matchings missing $C$.

For~\eqref{eq:strongpath}, write $g=g_a$, $r_i=r_i(a)$. The bounds $4,48$ suffice if $B_{ij}\le8$. Otherwise~\eqref{eq:bottleneck} and $a\ge\lambda$ give a simple $i$-to-$j$ path with edge weights at least $b:=B_{ij}>8$. Starting with $(S,x)=(U,i)$, move $x$ until it meets another original hole. We transfer the bound on $g(S)/r_x$ so intermediate terms cancel on summation. If the next vertex $z$ lies in $S$, \eqref{eq:edge-insertion} gives
\begin{equation}\label{eq:path-stop}
 \frac{g(S)}{r_x}\le
 \frac{g(S\setminus\{x,z\})}{br_x}\le\frac{32}{b}.
\end{equation}
If $z\notin S$, choose $p$ with $g(zp)=r_z$. Since $g(xz)\le1/a_{xz}<1/8\le r_z$, we have $p\ne x$. Apply~\eqref{eq:hole-exchange} with $U_1=S,U_2=\{z,p\}$, then~\eqref{eq:edge-insertion} to terms containing both $x,z$. Writing $S-x+w=(S\setminus\{x\})\cup\{w\}$ gives
\[
 g(S)\le\frac{r_x}{r_z}g(S-x+z)
 +\frac1{br_z}\left[\1_{\{p\notin S\}}g(S-x+p)
 +\sum_{y\in S\setminus\{x,p\}}g(S\setminus\{x,y\})g(py)\right].
\]
The endpoint $z$ gives the first term using $g(xp)\le r_x$. The bracket contains the term for $p$ when $p\notin S$ and those for $y\in S\setminus\{x,p\}$. For four holes the bracket is at most $48+3\cdot4^2=96$, and for two holes at most $4+4=8$. Since $r_x,r_z\ge1/8$,
\begin{equation}\label{eq:path-step}
 \frac{g(S)}{r_x}\le\frac{g(S-x+z)}{r_z}+\frac{6144}{b}.
\end{equation}
By~\eqref{eq:path-stop} and $r_i\le4$, at most $n-2$ steps of~\eqref{eq:path-step} give $g(U)\le4[6144(n-2)+32]/b\le10^5n/b$.
\end{proof}

\begin{samepage}
Following~\cite[proof of Theorem 3.1]{OpenAI}, we obtain the completion bound.

\begin{lemma}\label{lem:inflation}
If $\lambda$ is balanced and $\mu=\lambda+B/D$, with $D=10^8(n+1)^4$, then
\begin{equation}\label{eq:inflation}
 1\le Z_\mu/Z_\lambda\le e^{10^{-3}}<2.
\end{equation}
\end{lemma}
\end{samepage}
\begin{proof}
Set $a(t)=\lambda+tB/D$, $Z(t)=Z_{a(t)}$, and $g_t=g_{a(t)}$ for $0\le t\le1$. Differentiating, with sums over unordered edges, gives
\begin{equation*}
 \begin{aligned}
 (\log Z)'&=D^{-1}\sum_{ij}B_{ij}g_t(ij),\\
 g_t(xy)'&=D^{-1}\sum_{ij\text{ disjoint from }xy}B_{ij}g_t(xyij)
 -g_t(xy)(\log Z)'.
 \end{aligned}
\end{equation*}
While $r_i(a(t))\in[1/8,4]$ for every $i$, \cref{lem:holes} bounds each product $B_{ij}g_t(ij)$ or $B_{ij}g_t(xyij)$ by $10^5n$. Since $\binom n2 10^5n/D\le10^{-3}=:\eta$,
\[
 0\le(\log Z)'\le\eta,\qquad -\eta g_t(xy)\le g_t(xy)'\le\eta.
\]
Integration gives $e^{-\eta t}g_0(xy)\le g_t(xy)\le g_0(xy)+\eta t$, hence $\tfrac14e^{-\eta}\le r_i(a(t))\le1+\eta$. If any $r_i$ first reaches $1/8$ or $4$, these bounds still hold by continuity, keeping every $r_i$ strictly inside, a contradiction. The derivative bounds thus hold on $[0,1]$, and integration proves~\eqref{eq:inflation}.
\end{proof}

\subsection{The marked law}
Each edge $e$ can be labeled real for $\lambda_e$ or virtual for $B_e/D$ in $\mu_e$, following~\cite[Proposition 4.2]{OpenAI}. A marked matching $(M,V)$ has virtual edges $V\subseteq M$ and law
\begin{equation}\label{eq:marked-law}
 \Pi_\lambda(M,V)=\frac1{Z_\mu}
 \prod_{e\in M\setminus V}\lambda_e\prod_{e\in V}\frac{B_e}{D}.
\end{equation}
Summing over $V$ gives $\pi_\mu$. Balance ensures that the all-real event $V=\varnothing$ has probability $Z_\lambda/Z_\mu\ge1/2$, by \cref{lem:inflation}. Sampling needs no balance: combine \cref{thm:structural-sampler} with the marking rule of~\cite[Proposition 4.2]{OpenAI}.

\begin{theorem}\label{thm:marked-sampler}
Given even $n\ge4$, positive rational $\lambda$ with $\lambda_e\le4^Q$, and integers $Q,p\ge1$, an algorithm returns a marked matching with law within $2^{-p}$ of $\Pi_\lambda$ in total variation. Bit running time and fair-random-bit use are polynomial in $n,Q,\bit(\lambda),p$ on every run, regardless of balance.
\end{theorem}
\begin{proof}
Since $1\le B_e\le4^Q$ and $\lambda\le B$, the completion satisfies
\begin{equation*}
 D^{-1}\le\mu_e\le2\cdot4^Q,\qquad B/D\le\mu\le(D+1)B/D.
\end{equation*}
Apply \cref{thm:structural-sampler} to $\mu$ and independently mark each sampled edge $e$ virtual with probability $B_e/(D\mu_e)$. Internal scaling preserves $\pi_{D\mu}=\pi_\mu$, while marking uses the original weights. An exact sample gives~\eqref{eq:marked-law}. The rounding budget in \cref{app:simulation} includes the $n/2$ marking coins. Mixing and total rounding errors each at most $2^{-p-1}$ prove the claim.
\end{proof}

For balance in \cref{sec:counting}, use \emph{vertex scaling} $a_{ij}\mapsto a_{ij}s_i s_j$, with $s_i>0$. Every full matching weight gains $\prod_i s_i$, preserving its law and dividing $g_a(ij)$ by $s_i s_j$.

\section{Counting with one marked batch}\label{sec:counting}
We prove \cref{thm:main} using OpenAI's marked annealing method~\cite[Section 8]{OpenAI}. Each stage uses one batch to estimate the next partition-function ratio and restore balance.

\subsection{Annealing and marked samples}
For $n=0$, return one. For $n>2|E(G)|$, return zero. Otherwise, polynomial time in $n$ is polynomial in input length, even for binary $n$. Use Edmonds' algorithm~\cite{Edmonds} to test for a perfect matching, returning zero if none exists. Handle $n=2$ exactly. Henceforth $n\ge4$ is even and $Z(G)\ge1$.

Let $b=1+1/n$. The outer schedule assigns weights $w_e^{(j)}=1$ to graph edges and $b^{-j}$ to nonedges, with $Z_j=Z_{w^{(j)}}$. To suppress nonedges, take the least $\Jout\ge1$ satisfying
\begin{equation}\label{eq:annealing-length}
 n^{n/2}b^{-\Jout}\le\varepsilon/32.
\end{equation}
Since $\log b\ge1/(2n)$, we have $\Jout=O(n(n\log n+\log(1/\varepsilon)))$. The weight multiplier for a full or two-hole matching $H$ and the full partition-function ratio are
\begin{equation*}
 u(H)=b^{-|H\setminus E(G)|}\in[1/2,1],\qquad
 R_j=Z_{j+1}/Z_j\in[1/2,1].
\end{equation*}
These bounds follow from $b^{n/2}\le e^{1/2}<2$, and $Z_0=(n-1)!!$.

At outer stage $j$, scale the weights to $\lambda_{uv}=w_{uv}^{(j)}s_us_v$ with $s_i>0$. Sampling uses their completion $\mu=\lambda+B/D$ and the internal levels of \cref{sec:sampling}. Set $\lambda^+_{uv}=w_{uv}^{(j+1)}s_us_v$, advancing the outer schedule before rebalancing. Vertex scaling gives $\pi_\lambda=\pi_{w^{(j)}}$ and $Z_{\lambda^+}/Z_\lambda=R_j$. Initialize $s_i=1/f$, where $f=\lfloor\sqrt{n-1}\rfloor$, so initially every two-hole ratio is $f^2/(n-1)\in[1/4,1]$. The balance condition is
\begin{equation}\label{eq:invariant}
 \tfrac14\le r_i(\lambda)\le1\qquad\text{for every }i.
\end{equation}
Full and two-hole partition sums both change by factors in $[1/2,1]$. Under~\eqref{eq:invariant}, this gives
\begin{equation}\label{eq:prepass}
 \tfrac12g_\lambda(e)\le g_{\lambda^+}(e)\le2g_\lambda(e),
 \qquad \tfrac18\le r_i(\lambda^+)\le2.
\end{equation}
To estimate the next two-hole ratios, record from each marked sample $(M,V)$ at $\lambda$
\begin{equation*}
 a=\1_{\{V=\varnothing\}},\qquad d=a\,u(M),\qquad
 z_e=B_e^{-1}\1_{\{V=\{e\}\}}u(M\setminus\{e\}).
\end{equation*}
Here $a$ records all-real samples and $d$ reweights them to $\lambda^+$. Removing the unique virtual edge $e$ gives a bijection to perfect matchings after deleting its endpoints (written $-e$). In $z_e$, $B_e^{-1}$ leaves the factor $1/D$ from~\eqref{eq:marked-law}, and $u$ reweights the remaining matching. All three observables lie in $[0,1]$ since $B_e\ge1$, with exact means
\begin{equation}\label{eq:marked-means}
 \E a=\frac{Z_\lambda}{Z_\mu},\qquad
 \E d=\frac{Z_{\lambda^+}}{Z_\mu},\qquad
 \E z_e=\frac{Z_{\lambda^+}(-e)}{DZ_\mu}.
\end{equation}
Taking ratios yields
\begin{equation*}
 R_j=\frac{\E d}{\E a},\qquad
 g_{\lambda^+}(e)=D\frac{\E z_e}{\E d}.
\end{equation*}
By \cref{lem:inflation}, balance gives $\E a\ge1/2$ and $\E d\ge1/4$.

\subsection{Estimation and a dyadic balancing pass}
For $N_{\mathrm{obs}}=\binom n2+2$ observables and failure probability $\delta/\Jout$ per stage, set
\begin{equation}\label{eq:batch-parameters}
 \begin{aligned}
 \alpha&=2^{-12},&
 \sigma&=\min\left\{\frac{\alpha}{100D},
                         \frac{\varepsilon}{1000\Jout}\right\},\\
 L_\delta&=\left\lceil\log_2\frac{2\Jout N_{\mathrm{obs}}}{\delta}\right\rceil,&
 S&=\left\lceil2L_\delta\sigma^{-2}\right\rceil,\qquad
 p=\left\lceil\log_2\frac2\sigma\right\rceil.
 \end{aligned}
\end{equation}
Here $\sigma$ controls the ratios and final product, $S$ the failure probability, and $p$ the sampling bias. Use $Q=\Jout$ in all calls, as justified in \cref{app:countbits}.

At each stage $j=0,\ldots,\Jout-1$, form empirical means $\bar a,\bar d,\bar z_e$ from $S$ fresh independent marked samples at $\lambda$. Floors and clipping keep updates defined even when estimates fail:
\begin{equation*}
 \begin{gathered}
 a_* =\max(1/4,\bar a),\qquad d_* =\max(1/8,\bar d),\\
 \widehat R_j=\operatorname{clip}_{[1/2,1]}(\bar d/a_*),\qquad
 G_e=\min\{3,\alpha+D\bar z_e/d_*\},
 \end{gathered}
\end{equation*}
where $\operatorname{clip}_{[l,u]}(x)=\max(l,\min(u,x))$. Always $\widehat R_j\in[1/2,1]$ and $\alpha\le G_e\le3$.

To restore balance, adapt OpenAI's pass~\cite[Lemmas 8.1--8.2]{OpenAI} by rounding scale multipliers upward to powers of two. Start $v_i=2$. Visit the vertices once in a fixed order, replacing $v_i$ by the smallest member of
\begin{equation}\label{eq:grid-constraint}
 \mathcal H=\{2^{-13},2^{-12},\ldots,2^1\}
 \quad\text{such that}\quad v_iv_k\ge G_{ik}\quad(k\ne i).
\end{equation}
Then replace $s_i$ by $s_iv_i$ and advance to stage $j+1$.

Suppose $\lambda$ is balanced and each empirical mean differs from its exact mean in~\eqref{eq:marked-means} by at most $\sigma$. The denominator floors are inactive. Writing $g_e=g_{\lambda^+}(e)\le2$, we obtain
\begin{equation*}
 \left|\frac{D\bar z_e}{\bar d}-g_e\right|
 \le\frac{(D+g_e)\sigma}{\E d-\sigma}
 \le\frac{(D+2)\sigma}{1/4-\sigma}<\alpha.
\end{equation*}
The cap at three is inactive, and adding $\alpha$ makes $G_e$ an upper bound:
\begin{equation}\label{eq:data-sandwich}
 g_e\le G_e\le g_e+2\alpha.
\end{equation}
Clipping the partition-function ratio cannot increase its error relative to $R_j\in[1/2,1]$, so
\begin{equation}\label{eq:ratio-error}
 \left|\frac{\widehat R_j}{R_j}-1\right|
 \le\frac{2\sigma}{R_j(\E a-\sigma)}
 \le\frac{4\sigma}{1/2-\sigma}<10\sigma
 \le\frac{\varepsilon}{100\Jout}.
\end{equation}

\begin{lemma}\label{lem:balance-pass}
The pass is defined for all symmetric data $\alpha\le G_e\le3$. If~\eqref{eq:prepass} and~\eqref{eq:data-sandwich} hold, it restores~\eqref{eq:invariant}.
\end{lemma}
\begin{proof}
Initially $v_iv_k=4\ge G_{ik}$. The current $v_i$ is feasible, so updates only decrease multipliers and preserve the constraints. Vertex scaling divides each true ratio by $v_iv_k$, giving the upper bound one under~\eqref{eq:data-sandwich}.

For the lower bound, at vertex $i$ let $m_i=\max_{k\ne i}G_{ik}/v_k\ge\alpha/2=2^{-13}$, attained at $k_i$. Minimality gives $m_i\le v_i<2m_i$, hence $G_{ik_i}/(v_iv_{k_i})>1/2$. Later decreases cannot reduce this ratio. Under~\eqref{eq:prepass} and~\eqref{eq:data-sandwich}, $m_i\ge r_i(\lambda^+)/2\ge1/16$, so final products are at least $1/256$. The additive error $2\alpha$ on $ik_i$ therefore gives
\[
 \max_{k\ne i}\frac{g_{\lambda^+}(ik)}{v_iv_k}
 >\tfrac12-512\alpha=\tfrac38>\tfrac14.\qedhere
\]
\end{proof}

\subsection{Accuracy and termination}
Since the weights depend on earlier samples, condition on the history before each stage, fixing the activities and observables. By \cref{thm:marked-sampler}, independent calls have law within $2^{-p}$ of $\Pi_\lambda$, so each observable $X\in[0,1]$ has mean bias at most $2^{-p}\le\sigma/2$. Hoeffding's inequality~\cite{Hoeffding}, applied around the actual means, gives
\[
 \Pp\bigl(|\bar X-\E_{\Pi_\lambda}X|>\sigma
       \mid\text{history}\bigr)
 \le2e^{-S\sigma^2/2}.
\]
A union bound over the observables, allowing dependence within a sample, yields
\begin{equation*}
 \Pp(\text{some empirical error exceeds }\sigma\mid\text{history})
 \le2N_{\mathrm{obs}}e^{-S\sigma^2/2}
 \le2N_{\mathrm{obs}}2^{-L_\delta}
 \le\frac{\delta}{\Jout}.
\end{equation*}
Averaging over histories and summing over stages gives simultaneous empirical accuracy with probability at least $1-\delta$. On this event, \cref{lem:balance-pass} propagates balance by induction, and~\eqref{eq:ratio-error} controls every ratio. We return
\begin{equation*}
 \widehat Z=(n-1)!!\prod_{j=0}^{\Jout-1}\widehat R_j.
\end{equation*}
On the same event, multiplying~\eqref{eq:ratio-error} over stages gives
\begin{equation*}
 (1-\varepsilon/100)Z_{\Jout}\le\widehat Z
 \le e^{\varepsilon/100}Z_{\Jout}.
\end{equation*}

At most $n^{n/2}$ matchings use a nonedge, each with final weight at most $b^{-\Jout}$. Hence
\begin{equation*}
 Z(G)\le Z_{\Jout}\le Z(G)+\varepsilon/32
 \le(1+\varepsilon/32)Z(G).
\end{equation*}
For $\varepsilon<1$, $e^{\varepsilon/100}\le1+\varepsilon/50$, so the combined error is within $1\pm\varepsilon$. The deterministic polynomial cost bound in \cref{app:countbits} also holds when estimates fail.

\appendix
\section{Finite precision}\label[appendix]{app:implementation}
\subsection{Simulation of the replicated chain}\label{app:simulation}
There are at most $n^{n/2}$ perfect matchings, and the capped activities in~\eqref{eq:capped-path} lie in $[1,2D4^Q]$. Thus
$\pi_t(M)\ge[n^{n/2}(2D\cdot4^Q)^{n/2}]^{-1}$. With the parameters $T,q,J$ of \cref{sec:sampling}, set
\begin{equation*}
 \begin{aligned}
 \Lambda&=nq(T+1)(\lceil\log_2(2nD)\rceil+2Q+2),\\
 s(\xi)&=\left\lceil\frac{8J}{3}
 (\Lambda+\lceil\log_2(1/\xi)\rceil+1)\right\rceil
 \qquad(0<\xi<1).
 \end{aligned}
\end{equation*}
Here $\log(1/\min\nu)\le\Lambda$. For any initial $z$, $L^2$ contraction and Cauchy--Schwarz give
\begin{equation*}
 \norm{P^{s(\xi)}(z,\cdot)-\nu}_{\TV}
 \le\tfrac12\left(1-\tfrac3{8J}\right)^{s(\xi)}\sqrt{1/\nu(z)-1}\le\xi.
\end{equation*}
The same bound holds for the top-level output.

Precompute $c^t$ and $a^{(t)}$. Decode update indices by integer division, compute Metropolis ratios from edge products (partition functions cancel), and list discrepancy cycles when choosing a cycle exchange. A uniform refresh takes $n/2$ partner choices. By~\eqref{eq:adjacent-laws}, all exponents and arithmetic lengths are polynomial in $n,Q,\log D,\bit(a)$.

For a rational law $p_1,\ldots,p_h$, use $k$ fair bits to choose $Z$ uniformly in $\{0,\ldots,2^k-1\}$ and select the half-open cumulative interval containing $Z/2^k$. The total variation error is at most $2h2^{-k}$. A run, including $n/2$ final marking coins when used, has at most
\[
 R_{\mathrm{draw}}=100(s(\xi)+1)(n+1),\qquad
 R_{\mathrm{opt}}=10(J+n^2+1)
\]
choices and options per choice, respectively. Take the least $k$ with $2^k\ge2R_{\mathrm{draw}}R_{\mathrm{opt}}/\xi$. Conditional coupling over these choices bounds the total rounding error by $\xi$. Exact marking cannot increase total variation distance, so the unmarked or marked output has error at most $2\xi$ and uses at most $kR_{\mathrm{draw}}$ fair bits. With $\xi=2^{-b-1}$ or $2^{-p-1}$, this proves \cref{thm:structural-sampler,thm:marked-sampler}, respectively. Marking probabilities have polynomial bit length.

\subsection{The adaptive counter}\label{app:countbits}
Each multiplier in~\eqref{eq:grid-constraint} adds an integer in $[-13,1]$ to the scale exponent. Thus, after $j$ passes,
\begin{equation*}
 s_i=2^{h_i}/f,\qquad -13j\le h_i\le j,
 \qquad f=\lfloor\sqrt{n-1}\rfloor,
\end{equation*}
on every history. Activities are positive and at most $4^j/f^2\le4^{\Jout}$, validating $Q=\Jout$. $w_e^{(j)}$ and $2^{h_i}/f$ have $O(\Jout\log(n+1))$ bits, as does $B_e^{-1}$, since $B_e$ is one or an input activity.

Within a stage, the $u(H)$ share denominator $(n+1)^{n/2}$, and for fixed $e$, $z_e$ is zero or $B_e^{-1}u(H)$. Averaging adds $O(\log S)$ denominator bits, so means and floored ratios have polynomial bit length. The data only choose dyadic multipliers, so means' denominators never enter cumulative scales. In the final ratio product, the factors' bit lengths add.

There are exactly $\Jout S$ marked sampler calls and $n\Jout$ updates over the finite grid~\eqref{eq:grid-constraint}. The parameters in~\eqref{eq:annealing-length} and~\eqref{eq:batch-parameters}, and $f$, are computable by rational arithmetic, integer power comparisons and integer square root. In particular,
\[
 \sigma^{-1}=O(D+\Jout\varepsilon^{-1}),\qquad S=O\!\left(\sigma^{-2}[\log(n\Jout)+\log(1/\delta)]\right),\qquad p=O(\log(1/\sigma)).
\]
With $D=10^8(n+1)^4$, these bounds are polynomial in $n,\varepsilon^{-1},\log(1/\delta)$ and the input encodings. Summing the deterministic time and random-bit budgets of \cref{thm:marked-sampler} over the fixed calls proves the worst-case bounds, without requiring balance on failed histories.

\end{document}